\documentclass[leqno]{amsart}
\usepackage{amsfonts,amssymb,amsmath,amsgen,amsthm}
\usepackage{comment}
\usepackage{hyperref,color}

\theoremstyle{plain}
\newtheorem{theorem}{Theorem}[section]

\newtheorem{lemma}[theorem]{Lemma}
\newtheorem{corollary}[theorem]{Corollary}
\newtheorem{proposition}[theorem]{Proposition}

\theoremstyle{remark}
\newtheorem{remark}[theorem]{Remark}

\newtheorem{example}[theorem]{Example}

\def\dis{\displaystyle}

\def\C{{\mathbf C}}% complex numbers
\def\R{{\mathbf R}}% real numbers
\def\N{{\mathbf N}}% nonnegative integers
\def\Sch{{\mathcal S}}% Schwartz space
\def\O{\mathcal O}
\def\F{\mathcal F}
\def\dd{\mathrm d}

\def\({\left(}
\def\){\right)}
\def\<{\left\langle}
\def\>{\right\rangle}
\def\le{\leqslant}
\def\ge{\geqslant}

\def\Tend#1#2{\mathop{\longrightarrow}\limits_{#1\rightarrow#2}}

\def\d{{\partial}}
\def\eps{\varepsilon}

\def\om{\omega}

\def\pp{p}

\numberwithin{equation}{section}

\begin{document}

\title{Higher-order Schr\"{o}dinger and Hartree--Fock equations}

\author[R. Carles]{R\'emi Carles}
\address{CNRS \& Univ. Montpellier\\ UMR5149\\ Math\'ematiques
\\CC051\\34095 Montpellier\\ France}
\email{Remi.Carles@math.cnrs.fr}

\author[W.~Lucha]{Wolfgang Lucha}
\address{Institute for High Energy Physics\\
Austrian Academy of Sciences\\ Nikolsdorfergasse 18, A-1050
Vienna, Austria} \email{Wolfgang.Lucha@oeaw.ac.at}

\author[E. Moulay]{Emmanuel Moulay}
\address{XLIM (UMR-CNRS 7252), Univ. Poitiers\\
11 bd Marie et Pierre Curie, BP 30179\\
86962 Futuroscope Chasseneuil Cedex\\ France}
\email{emmanuel.moulay@univ-poitiers.fr}

\begin{abstract}
The domain of validity of the higher-order Schr\"{o}dinger
equations~is analyzed for harmonic-oscillator and Coulomb
potentials as typical examples. Then the Cauchy theory for
higher-order Hartree--Fock equations with~bounded and Coulomb
potentials is developed. Finally, the existence of associated ground
states for the odd-order equations is proved. This renders these
quantum equations relevant for physics.
\end{abstract}

\thanks{RC was supported by the French ANR projects SchEq
 (ANR-12-JS01-0005-01) and BECASIM (ANR-12-MONU-0007-04).}
\maketitle

\section{Introduction}
In this work, we discuss, in a rather general setting, various
equations of motion enjoying considerable interest in numerous areas
of physics and, by establishing the well-posedness of the
corresponding problems, we try to put applications of these equations
of motion on a solid basis:
The higher-order Schr\"{o}dinger equations have been developed in, e.g.,
\cite{CaMo12,Kim12}. These~are Schr\"{o}dinger-type equations
involving a higher-order Schr\"{o}dinger operator
\cite{Cycon87,Gorban07,Helffer88,Karpeshina12} and converging
towards the semirelativistic bound-state equation called~the
spinless Salpeter equation. The original Schr\"{o}dinger equation
has been formulated by Schr\"{o}dinger in 1926
\cite{Schrodinger26}. The spinless Salpeter equation is studied,
for instance, in \cite{Hall05,Hall07}. The Cauchy problem of the
higher-order Schr\"{o}dinger equations without potential, i.e.,
for free particles, is studied in \cite{CaMo12,Kim12}. The case of
bounded potentials (e.g., particles in finite potential wells) and
of linear potentials (e.g., neutrons in free fall in the gravity
field and electrons accelerated by an electric field) is treated
in~\cite{CaMo12}. Moreover, the higher-order Schr\"{o}dinger
operator with quasi-periodic potentials in two dimensions is
discussed in \cite{Karpeshina12}.

The extraordinarily high interest of the particle physics
community in the spinless Salpeter equation derives primarily from
the fact that this semirelativistic equation of motion constitutes
a well-defined approximation to the Bethe--Salpeter formalism
\cite{Salpeter51} designed for the Lorentz-covariant
description of bound states within~relativistic quantum field
theory. Within this framework, it may be obtained along the
course~of a ``three-dimensional reduction'' effected by a sequence
of reasonable and physically justified assumptions \cite{Lucha91,Lucha99}: In the limit of all bound-state
constituents interacting instantaneously \cite{Lucha05} as
well as propagating freely, the homogeneous Bethe--Salpeter
equation reduces to Salpeter's equation \cite{Salpeter52},
which, upon neglect of negative-energy contributions and spin
degrees of freedom and restriction of the involved interaction
kernels to convolution form, eventually simplifies to the spinless
Salpeter equation. The latter may be regarded as the
straightforward generalization of the Schr\"odinger equation
towards inclusion of the relativistically correct free-particle
kinetic energy.

The semirelativistic bound-state equation emerging from the
derivation sketched above is the eigenvalue equation of a nonlocal
Hamiltonian composed of the sums~of the relativistic kinetic
energies of the bound-state constituents --- represented by~the
famous square-root operators --- and of the interactions of these
particles --- encoded in appropriately chosen potentials. The
inconvenience induced by the nonlocality~of such a
spinless-Salpeter Hamiltonian, however, occasionally tempts
practitioners~to expand each square-root operator, regarded as a
function of the involved momentum squared, in a (truncated) Taylor
series not only up to the lowest nontrivial order~--- which gives
the usual Schr\"odinger equation --- but (at least) up to the
next-to-lowest nontrivial order, without paying attention to
obviously crucial questions such as~the well-posedness~of the
problem or the existence of a ground state. Unfortunately,~the
Hamiltonian involving this expansion up to next-to-lowest
nontrivial order proves~to be unbounded from below. In this
situation, a remedy \emph{might\/} be, \emph{if done properly\/},
to construct merely approximate solutions to such \emph{pseudo\/}
spinless Salpeter equation~by taking into account all the
worrisome terms in the Hamiltonian only perturbatively.

In order to settle this question once and forever, we analyze such
semirelativistic equations of motion for expansions of the
relativistic kinetic energy up to arbitrarily high order. The
outcomes of such expansions are known as higher-order
Schr\"odinger equations.

Logically, the next steps then are to consider the corresponding
time-dependent equations, to allow for the general case of systems
composed of more than one~or~two interacting particles, which
yields the Hartree equation, and to take into account the
fermionic nature of the involved particles, which leads to the
Hartree--Fock equation.

The Hartree equation, found by Hartree in the 1920s
\cite{Fischer03}, arises in the mean-field limit of large systems
of identical bosons (as, e.g., the Gross--Pitaevskii equation~for
Bose-Einstein condensates \cite{Gross61,Pitaevskii61}) by taking
into account the self-interactions~of~the bosons. A
semirelativistic version of the Hartree equation was obtained in
\cite{Elgart07,Lenzmann07} for modeling boson stars. The
Hartree--Fock equation, also developed by Fock~\cite{Fock30},
describes large systems of identical fermions (finding application
in, e.g., electronic structure theory) by taking into account the
self-interactions of charged fermions as well as an exchange term
resulting from Pauli's principle. A semirelativistic version of
the Hartree--Fock equation was developed in \cite{Frohlich07} for
modeling white dwarfs. The Hartree equation is also used for
fermions as an approximation of the Hartree--Fock equation
neglecting the impact of their fermionic nature. Hartree and
Hartree--Fock equations are used for several applications in
many-particle physics \cite[Section 2.2]{Lipparini08}.

Our first issue is to extend the scope of the higher-order
Schr\"{o}dinger equations~to the Coulomb potential by using
perturbation theory. Coulomb potentials have~been widely employed
for the Schr\"{o}dinger equation (see, for instance,
\cite{Greene76,Lam71,Rogers70,Rau71,Yajima87}); among others, it
is conceivable to apply the higher-order Schr\"{o}dinger
equations~with a Coulomb potential to the $\alpha$ particles,
semirelativistically moving charged bosons composed of two protons
and two neutrons and produced in the nuclear $\alpha$
decay~\cite{Lilley01}. Generalizing a result obtained in
\cite{CaMo12}, we also prove that the higher-order Schr\"{o}dinger
equations converge towards the spinless Salpeter equation. Our
second issue is to develop higher-order Hartree--Fock equations
for bounded and Coulomb potentials. This allows us to take into
account some relativistic effects in many-particle physics, as,
for instance, the electrons of heavy atoms in quantum chemistry
\cite{Kaldor03,Pyykko77,Szabo96}, the semirelativistic electron
gas in a finite potential well \cite[Section 4.2]{Martin04}, the
metal clusters \cite[Section 2.2.1]{Lipparini08} or the modeling
of white dwarfs \cite{Frohlich07}. Last but not least,~in order to
give a physical meaning to the higher-order Schr\"{o}dinger and
Hartree--Fock equations, the existence of a ground state is proved
for the odd-order~equations.

This paper is organized as follows. After recalling some notations
and definitions in Section~\ref{Sec Def}, the case of higher-order
Schr\"{o}dinger equations with harmonic-oscillator or Coulomb
potentials is discussed in Section~\ref{sec:linear}.
Section~\ref{Sec HF BP} is devoted to the Cauchy problem of the
higher-order Hartree--Fock equations. An important special case,
the Cauchy problem of higher-order Hartree--Fock equations with a
Coulomb potential, is addressed in Subsection~\ref{Sec HF
Coulomb}.
%Numerical simulations
 % showing the convergence of the model to the semi-relativistic regime
 % are presented in Section~\ref{sec:num}.
In Section~\ref{sec:ground}, we prove the existence of a
ground~state for the odd-order Schr\"{o}dinger and Hartree--Fock
equations. Our conclusions may be found in Section~\ref{Sec
Conclusion}. In the Appendix, an extension of the convergence of
the higher-order Schr\"{o}dinger equations towards the spinless
Salpeter equation is provided.

\section{Notations and definitions}\label{Sec Def}
The wave function of a particle is denoted by $\psi(t,x)$, where
$x$ is the position of the particle and $t$ the time. Moreover,
$\psi$ stands for $\psi(t,x)$. $\Delta:=\nabla^2$ denotes the
Laplace operator.

For $x\in\R^3$, $|x|$ denotes the Euclidean norm of $x$. The
notation $\ast$ stands for the convolution, defined, in the case of integrable
functions, by the formula
\begin{equation}
\left(f\ast g\right)(x)=\int_{\R^3} f(x-y)g(y)dy.
\end{equation}
%The adjoint of a complex vector $z\in\C^3$ is defined by $z^*=\overline{z}^T$.

Upon approximating the relativistically correct expressions for
the kinetic energy
\begin{equation}\label{Energy}
E=\sqrt{\textbf{p}^2 c^2 + m^2 c^4}
\end{equation}
of a free particle of mass $m$ and momentum $\textbf{p}$ by its
expansion in powers of $\textbf{p}^2/m^2$,
\begin{equation}\label{eq Taylor approx}
E_J=m c^2 \left(1+\sum_{j=1}^{J}(-1)^{j+1} \alpha(j)
\frac{\textbf{p}^{2j}}{m^{2j} c^{2j}}\right),
\end{equation}
and applying the correspondence principle \cite{Bohr76}
\begin{equation}\label{Corresp}
E \leftrightarrow i\hbar \frac{\partial}{\partial t} \qquad \textbf{p}
\leftrightarrow -i\hbar \frac{\partial}{\partial x } = -i\hbar \nabla,
\end{equation}
we obtain the \emph{higher-order Schr\"{o}dinger equations}
\begin{equation}\label{eq:pot}
 i\hbar \frac{\partial \psi}{\partial t} =-\sum_{j=0}^J\frac{\alpha(j)\hbar^{2j}}{m^{2j-1}c^{2j-2}} \Delta^j\psi+ V\psi,
\end{equation}
where $V$ is an external potential, $J\in\N^*$, $\hbar=\frac{h}{2
\pi}$ the reduced Planck constant, $c$ the speed of light, and
\begin{equation}\label{eq:alpha}
\alpha(j)=\frac{(2j-2)!}{j!(j-1)!2^{2j-1}},\quad j\ge 1,
\end{equation}
with $\alpha(0)=-1$. Stirling's formula yields, in particular,
\begin{equation}\label{eq:stirling}
 \alpha(j)=\O\(\frac{1}{j^{3/2}}\) \quad\text{as }j\to \infty.
\end{equation}
See \cite{CaMo12} for more details. We denote the higher-order
kinetic-energy operator by
\begin{equation}\label{eq:H0J}
\mathcal{H}_{0J}=-\sum_{j=0}^J
\frac{\alpha(j)\hbar^{2j}}{m^{2j-1}c^{2j-2}} \Delta^j.
\end{equation}
\begin{example}\label{ex:H0J}
For $J=1$, we get the regular Schr\"odinger operator
 \begin{equation*}
 \mathcal{H}_{01}= mc^2 - \frac{\hbar^2}{2m}\Delta.
 \end{equation*}
For $J=2$, we get
\begin{equation*}
 \mathcal{H}_{02}= mc^2 - \frac{\hbar^2}{2m}\Delta
 -\frac{\hbar^4}{8 m^3 c^2}\Delta^2.
\end{equation*}
\end{example}

The semirelativistic \emph{time-dependent spinless Salpeter
equation} is given by
\begin{equation}\label{SR}
i \hbar \frac{\partial}{\partial t}\psi=\sqrt{-c^2 \hbar^2 \Delta + m^2 c^4} \ \psi+ V \psi.
\end{equation}

Let us introduce the integro-differential Hartree and
Hartree--Fock equations given, for instance, in
\cite{Frohlich07,Lipparini08,Martin04}. The \emph{Hartree
equation} of $N$ particles is defined~by
\begin{equation}\label{H}
i \hbar \frac{\partial}{\partial t}\psi_k=-\frac{\hbar^2}{2m}
\Delta \psi_k+\sum_{\ell=1}^N \left(\frac{\kappa}{|x|}\ast
\left|\psi_\ell\right|^2\right)\psi_k + V \psi_k,
\end{equation}
with $V$ an external bounded potential, $\kappa$ a real constant and
$k=1,\ldots,N$. The Hartree factor
\begin{equation*}
H=\sum_{\ell=1}^N \left(\frac{\kappa}{|x|}\ast
 \left|\psi_\ell\right|^2\right)
\end{equation*}
describes the self-interaction between charged particles as a
repulsive force if $\kappa >0$, attractive force if $\kappa<0$. The
\emph{Hartree--Fock equation} of $N$
particles is given by
\begin{equation}\label{HF}
i \hbar \frac{\partial}{\partial t}\psi_k=-\frac{\hbar^2}{2m} \Delta
\psi_k+H \psi_k-\sum_{\ell=1}^N \left(\frac{\kappa}{|x|}\ast \overline{\psi_\ell} \psi_k \right)\psi_\ell + V \psi_k.
\end{equation}
The Fock term
\begin{equation}\label{eq:Fock}
F_k(\psi_k)=\sum_{\underset{\ell \neq k}{\ell=1}}^N
\(\frac{\kappa}{|x|}\ast (\overline{\psi_\ell} \psi_k)\)\psi_\ell
\end{equation}
is an exchange term that is a consequence of the Pauli principle
and thus applies~to fermions.

The \emph{semirelativistic Hartree equation} is given by
\begin{equation}\label{RH}
i \hbar \frac{\partial}{\partial t}\psi_k=\sqrt{-c^2 \hbar^2 \Delta +
 m^2 c^4} \ \psi_k+H \psi_k+ V \psi_k,
\end{equation}
and the \emph{semirelativistic Hartree--Fock equation} by
\begin{equation}\label{RHF}
i \hbar \frac{\partial}{\partial t}\psi_k=\sqrt{-c^2 \hbar^2 \Delta +
 m^2 c^4} \ \psi_k+H \psi_k-F_k(\psi_k)+ V \psi_k,
\end{equation}
with $k=1,\ldots,N$. Both equations \eqref{RH} and \eqref{RHF} are
studied, for instance, in \cite{Aki08,Cho09,Frohlich07}. The main
difficulty of all these equations relies in the use of the
nonlocal pseudo-differential operator $\sqrt{-c^2 \hbar^2\Delta +
m^2 c^4}$ (see, e.g., \cite[Chapter 7]{LiebLoss} or
\cite{CHHO13}).
%Indeed, it causes problem
%for practical simulations on a bounded domain where an infinite number
%of boundary conditions are necessary.

For $J\in\N^*$, we have the following \emph{higher-order
Hartree--Fock equations}:
\begin{equation}\label{HOHF}
i \hbar \frac{\partial}{\partial t} \psi_k = {\mathcal
H}_{0J}\psi_k +H \psi_k-F_k(\psi_k)+ V \psi_k,
\end{equation}
with $k=1,\ldots,N$ and $V$ an external potential.

\section{Higher-order Schr\"{o}dinger equations with an external
 potential}\label{sec:linear}

It is proved in \cite{CaMo12} that equations \eqref{eq:pot} have a
unique solution without external potential, and for a bounded or
linear (in $x$) potential $V$. On the other hand, for a harmonic
potential, the flow associated to \eqref{eq:pot} is not well-defined
for $J=2$.

The first main objective of this section is to prove that if one
sticks to \emph{odd} values of $J$, then the flow associated to
\eqref{eq:pot} is well-defined with $V$ a harmonic-oscillator
potential. Then we consider the presence of a Coulomb potential.

\subsection{Harmonic-oscillator potential}
\label{sec:harmonic-potential}

Introduce the Schwartz space
\begin{equation*}
 \Sch(\R^d) = \left\{ f\in C^\infty(\R^d;\C)\ | \ \sup_{x\in
 \R^d}\left\lvert x^\alpha \d^\beta f(x)\right\rvert<\infty,\quad
 \forall \alpha,\beta\in \N^d\right\}.
\end{equation*}
For $f\in \Sch(\R^d)$, the (semi-classical) Fourier transform of $f$, denoted by
$\widehat f$ or $\F(f)$, is defined by
\begin{equation*}
 \widehat f (\pp) = \frac{1}{(2\pi\hbar )^{d/2}}\int_{x\in
 \R^d} e^{-ix\cdot
 \pp/\hbar} f(x)\dd x,
\end{equation*}
where $\pp$ is the Fourier variable. The Fourier Inversion Formula reads
\begin{equation}
 \label{eq:fif}
 f(x) = \frac{1}{(2\pi\hbar )^{d/2}}\int_{\pp\in \R^d} e^{+i x\cdot
 \pp/\hbar} \widehat f(\pp)\dd \pp.
\end{equation}
The Fourier transform is uniquely continuously extended to the space
of tempered distributions, $\Sch'(\R^d)$, and is unitary on
$L^2(\R^d)$ (Plancherel formula):
\begin{equation}
 \label{eq:plancherel}
 \|u\|_{L^2(\R^d)}=\|\widehat u\|_{L^2(\R^d)},\quad \forall u\in
 L^2(\R^d).
\end{equation}
Among other features, the Fourier transform exchanges
differentiation and multiplication by a polynomial, typically
\begin{equation*}
 \F (-\hbar^2\Delta u)(\pp)=|\pp|^2\widehat u(\pp),
\end{equation*}
and
\begin{equation*}
 \F\( \mathcal{H}_{0J}u\)(\pp) = -\sum_{j=0}^J
\frac{\alpha(j)}{m^{2j-1}c^{2j-2}} (-1)^j|\pp|^{2j} \widehat u(\pp).
\end{equation*}
In the presence of a harmonic-oscillator potential, we get
\begin{equation*}
 \F\( \(\mathcal{H}_{0J}+\frac{|x|^2}{2}\)u\)(\pp) =
 \( -\frac{\hbar^2}{2}\Delta_\pp -\sum_{j=0}^J
\frac{\alpha(j)}{m^{2j-1}c^{2j-2}} (-1)^j|\pp|^{2j}\)\widehat u(\pp).
\end{equation*}
If $J=1$, we see that the Fourier transform maps the harmonic
oscillator to another harmonic oscillator, in agreement with
Example~\ref{ex:H0J}. For $J=2$, we find, still in~view of
Example~\ref{ex:H0J},
\begin{equation*}
 \F\( \(\mathcal{H}_{02}+\frac{|x|^2}{2}\)u\)(\pp) =
 \( -\frac{\hbar^2}{2}\Delta_\pp +mc^2 +\frac{|\pp|^2}{2m}
 -\frac{|\pp|^4}{8m^3c^2}\)\widehat u(\pp).
\end{equation*}
Since the operator on the right-hand side is not essentially
self-adjoint on $C_0^\infty(\R^d)$ (see, e.g., \cite{Dunford}),
the flow associated to the operator
$\mathcal{H}_{02}+\frac{|x|^2}{2}$ is not well-defined. This is so
essentially because trajectories associated to the Hamiltonian on
the right-hand side may reach an infinite speed, due to the fact that the
potential (in $\pp$) goes to $-\infty$ faster that quadratic. Such
a feature is ruled out if $J$ is restricted to odd~values, $J=
2n+1$, $n\in \N$, since then
\begin{equation*}
 \F\( \(\mathcal{H}_{0J}+\frac{|x|^2}{2}\)u\)(\pp) =
 \( -\frac{\hbar^2}{2}\Delta_\pp -\sum_{j=0}^{2n+1}
\frac{\alpha(j)}{m^{2j-1}c^{2j-2}} (-1)^j|\pp|^{2j}\)\widehat u(\pp),
\end{equation*}
and the potential in $\pp$ behaves for large $\pp$ as
\begin{equation*}
 +\frac{\alpha(2n+1)}{m^{4n+1}c^{4n}}|\pp|^{4n+2}.
\end{equation*}
This potential is therefore uniformly bounded from below (going to
$+\infty$ at infinity, hence it is confining), and the flow
associated to the corresponding Hamiltonian is well-defined (see,
e.g., \cite{ReedSimon2}):
\begin{proposition}
 Let $\psi_0\in L^2(\R^3)$ and $J$ be an \emph{odd} integer. Then
 $\mathcal{H}_{0J}$ is essentially self-adjoint, and  the Cauchy problem
 \begin{equation*}
 i\hbar \frac{\partial \psi}{\partial t} =\mathcal{H}_{0J}\psi+\frac{|x|^2}{2}
 \psi,\quad
\psi(x,0)=\psi_0(x),
 \end{equation*}
 has a unique
solution $\psi\in C(\R;L^2(\R^3))$. In addition, the following
conservation law holds:
\begin{equation*}
 \frac{d}{dt}\|\psi(t)\|_{L^2(\R^3)}^2=0.
\end{equation*}
\end{proposition}
\begin{remark}
 As noticed in \cite{CaMo12} and recalled above, the assumption that
 $J$ is odd is sharp.
\end{remark}

\subsection{Coulomb potential}\label{Sec Schrodinger Coulomb}

For $J\in\N^*$, let us consider the higher-order Schr\"{o}dinger
equations
\begin{equation}\label{eq:coulomb}
 i\hbar \frac{\partial \psi}{\partial t} =\mathcal{H}_{0J}\psi+ V_\alpha
 \psi,\quad
\psi(x,0)=\psi_0(x),
\end{equation}
where
\begin{equation}\label{eq:coulombpot}
V_\alpha(x)=\frac{\alpha}{|x|}
\end{equation}
is the --- attractive or repulsive --- Coulomb potential with
coupling constant $\alpha\in\R$. We can prove the existence of a
solution of Equation \eqref{eq:coulomb} by adopting perturbative
arguments based on the Kato--Rellich theorem.

\begin{theorem}\label{theo:coulomb}
Let $\psi_0\in L^2(\R^3)$. Equation~\eqref{eq:coulomb} has a
unique solution $\psi\in C(\R;L^2(\R^3))$, given by
$\psi(x,t)=e^{i\frac{t}{\hbar} \mathcal{H}_J} \psi_0(x)$, where
\begin{equation}
\mathcal{H}_J=\mathcal{H}_{0J}+ V_\alpha=-\sum_{j=0}^J
\frac{\alpha(j)\hbar^{2j}}{m^{2j-1}c^{2j-2}}
\Delta^j + V_\alpha.
\end{equation}
In addition, the following conservation law holds:
\begin{equation*}
 \frac{d}{dt}\|\psi(t)\|_{L^2(\R^3)}^2=0.
\end{equation*}
Finally, the same conclusions hold if $J$ is odd and
$\mathcal{H}_{0J}$ is replaced by
\begin{equation*}
 \mathcal{H}_{0J}+\frac{|x|^2}{2}.
\end{equation*}
\end{theorem}

\begin{proof}
The proof relies on the Kato--Rellich Theorem. By using $2J$
integrations by parts, it is easy to prove that the Hamiltonian
$\mathcal{H}_{0J}$, defined in \eqref{eq:H0J}, is essentially
self-adjoint. We will show that $V_\alpha$ is $\mathcal
H_{0J}$-bounded:
\begin{equation}\label{eq:Aborne}
\|V_\alpha \psi\|_{L^2(\R^3)}\le a \|\mathcal{H}_{0J}\psi\|_{L^2(\R^3)}+b
\|\psi\|_{L^2(\R^3)},\quad \forall \psi\in \Sch(\R^3),
\end{equation}
with $a<1$, and where $\Sch$ denotes the Schwartz space. By using
the Kato--Rellich Theorem (see, e.g., \cite[Theorem
X.12]{ReedSimon2}), we deduce that the Hamiltonian
\begin{equation*}
\mathcal{H}_J=\mathcal{H}_{0J}+V_\alpha
\end{equation*}
is a self-adjoint operator and $D(\mathcal{H}_J)=W^{J,2}(\R^3)=
H^J(\R^3)$. Then, we apply the Stone Theorem (see, e.g.,
\cite{ReedSimon2}) to conclude that for $\psi_0\in L^2(\R^3)$, the
Schr\"{o}dinger-type equation
\begin{equation*}
 i\hbar \frac{\partial \psi}{\partial t} =\mathcal{H}_J\psi
\end{equation*}
has a unique solution given by $\psi(x,t)=e^{i\frac{t}{\hbar}
 \mathcal{H}_J} \psi_0(x)$, and
$\|\psi(t)\|_{L^2}=\|\psi_0\|_{L^2}$ for all time $t$. We are left
with the proof of \eqref{eq:Aborne}. From the Hardy inequality
(see, e.g., \cite{BCD11}), there exists $C>0$ such that
\begin{equation*}
 \|V_\alpha \psi\|_{L^2(\R^3)}\le C\|\nabla \psi\|_{L^2(\R^3)}.
\end{equation*}
An integration by parts and the Cauchy--Schwarz inequality yield
\begin{equation*}
 \|V_\alpha \psi\|_{L^2(\R^3)}\le
 C\|\psi\|_{L^2(\R^3)}^{1/2}\|\Delta\psi\|_{L^2(\R^3)}^{1/2}\le
 \frac{C}{\eps}\|\psi\|_{L^2(\R^3)}+C\eps \|\Delta\psi\|_{L^2(\R^3)},
\end{equation*}
where we have used the Young inequality $2ab\le a^2+b^2$ for the
last estimate, and $\eps>0$ is to be fixed later. By considering
$E_J$ as a polynomial in $\textbf{p}$, and distinguishing the
small values of $\textbf{p}$ from the large values of
$\textbf{p}$, we readily check that there exist $C_J>0$ such that
\begin{equation*}
 \textbf{p}^2 \le C_J\(1+ E_J^2\),
\end{equation*}
skipping here irrelevant parameters $m$ and $c$. Using the
Plancherel identity, we infer
\begin{equation}\label{eq:1750}
 \|\Delta \psi\|_{L^2(\R^3)}\le C_J\( \|\psi\|_{L^2(\R^3)}
 +\|\mathcal H_{0J} \psi\|_{L^2(\R^3)}\),\quad \forall \psi\in
 \Sch(\R^3).
\end{equation}
Gathering all the estimates together, we obtain
\begin{equation*}
 \|V_\alpha \psi\|_{L^2(\R^3)}\le
\(\frac{C}{\eps}+ C_JC\eps\) \|\psi\|_{L^2(\R^3)}
 +C_JC\eps \|\mathcal H_{0J} \psi\|_{L^2(\R^3)},\quad \forall \psi\in
 \Sch(\R^3).
\end{equation*}
Choosing $\eps$ sufficiently small, we conclude that $V_\alpha$ is
$\mathcal H_{0J}$-bounded, with a relative bound $a<1$. The case
of a harmonic-oscillator potential follows the same line. If $J$
is odd, then, as noticed in Section~\ref{sec:harmonic-potential},
the symbol of the operator $\mathcal{H}_{0J}+\frac{|x|^2}{2}$ is
bounded from below, and so an inequality analogous to
\eqref{eq:1750} holds:
\begin{equation*}
 \|\Delta \psi\|_{L^2(\R^3)}\le C_J\( \|\psi\|_{L^2(\R^3)}
 +\left\|\(\mathcal H_{0J}
 +\frac{|x|^2}{2}\)\psi\right\|_{L^2(\R^3)}\),
\quad \forall \psi\in \Sch(\R^3),
\end{equation*}
provided that $J$ is odd. We can then conclude as above.
\end{proof}

\section{Higher-order Hartree--Fock equations}\label{Sec HF BP}
\label{sec:cauchyB}

\subsection{Bounded external potential}
\label{sec:boundedpot}
In this section, we study the Cauchy problem associated to
\eqref{HOHF}, in the case where $V$ is a bounded potential.
We denote by
\begin{equation*}
 U_J(t) = e^{-it {\mathcal H}_{0J}}
\end{equation*}
the propagator corresponding to the case $H=F_k=V=0$, where we recall
that ${\mathcal H}_{0J}$ is defined in \eqref{eq:H0J}. Given
$\psi_{01},\dots,\psi_{0N}\in L^2(\R^3)$, we rewrite the Cauchy
problem \eqref{HOHF} with $\psi_{k\mid t=0}=\psi_{0k}$ in an integral
form (Duhamel's principle): for $k=1,\dots, N$,
\begin{equation}
 \label{eq:duhamel1}
 \begin{aligned}
 \psi_k(t) &= U_J(t)\psi_{0k} - i\int_0^t U_J(t-s)\( H
 \psi_k\)(s)ds \\
&\quad+i\int_0^t U_J(t-s)\( F_k(
 \psi_k)\)(s)ds - i\int_0^t U_J(t-s)\( V
 \psi_k\)(s)ds.
\end{aligned}
\end{equation}
We prove the global existence of a unique solution to
\eqref{eq:duhamel1} with initial data in $L^2(\R^3)$, thanks to
dispersive estimates for $U_J$. The corresponding argument is
presented in the case of a bounded potential, and we show in
Section~\ref{Sec HF Coulomb} how it can be adapted to the case of
a Coulomb potential. At the end of Section~\ref{sec:boundedpot},
we show that if the initial data belong to $H^J(\R^3)$ (which, of
course, as $J$ increases, is a stronger and stronger requirement)
and the external potential is sufficiently smooth, then global
existence of a unique solution can be established with more basic
tools than dispersive estimates (an approach which does not seem
to be easily extended to the case of a Coulomb potential though).
See Proposition~\ref{prop:algebre}. \smallbreak

In the case $J=1$ (the Hartree--Fock equation), the existence and
uniqueness~of a solution has been established in \cite{ChGl75}
(see also \cite{CaLB99} for a proof using more recent~tools).
Therefore, we shall focus our presentation on the case $J\ge2$. We
emphasize a difference with the previous results: for $J=1$, the
$H^2$-regularity of the solution to \eqref{HF} is proven by
showing that $\psi_k$ and $\d_t \psi_k$ belong to
$C([0,T];L^2(\R^3))$ and using equation \eqref{HF} to infer that
$\Delta \psi_k \in C([0,T];L^2(\R^3))$, hence $\psi_k \in
C([0,T];H^2(\R^3))$. This is so even in the linear case, see the
proof of Lemma~2.1 in \cite{Yajima87} (a property which is used
also in \cite{CaLB99}). In the case of \eqref{HOHF}, this method
can be adapted to pass from an $L^2$-regularity to an
$H^{2J}$-regularity: this approach will be followed to prove
Theorem~\ref{theo:r3coulomb} below (case of a Coulomb potential).

\begin{theorem}\label{theo:cauchyNL1}
 Let $J\ge 2$, $V\in L^\infty(\R^3)$, $\psi_{01},\dots,\psi_{0N}\in
 L^2(\R^3)$. Then \eqref{eq:duhamel1} has a unique, global, solution
 \begin{equation*}
 (\psi_1,\dots,\psi_N)\in \(C(\R;L^2(\R^3))\cap L^{4J/3}_{\rm
 loc}\(\R;L^\infty(\R^3)\)\)^N.
 \end{equation*}
In addition, the following conservation laws hold, for all
$\ell,k\in\{1,\dots,N\}$:
\begin{equation}\label{eq:conservL2}
 \frac{d}{dt}\int_{\R^3}\overline{\psi_\ell} (t,x)\psi_k(t,x)dx=0.
\end{equation}
\end{theorem}
This result will follow from Lemma~\ref{lem:local1} in
Subsection~\ref{sec:fixed} below.
\begin{remark}
 The space $L^{4J/3}_{\rm loc}\(\R;L^\infty(\R^3)\)$ is mentioned in
 order to guarantee uniqueness. Other spaces based on the
 Strichartz-type estimates presented below would do the job as well.
\end{remark}

\subsubsection{Dispersive estimates and consequences}
\label{sec:disp}

From \cite[Theorem~4.1]{Kim12}, we have the following
local-in-time dispersive estimate. There exists $C>0$ such that
\begin{equation*}
 \|U_J(t)\|_{L^1(\R^3)\to L^\infty(\R^3)} \le
 \frac{C}{|t|^{3/(2J)}},\quad 0<|t|\le 1.
\end{equation*}
Formally, this estimate is the same as the one associated to the usual
Schr\"odinger group $e^{it\Delta}$ ($J=1$) on $\R^n$, with
$n=3/J$. This remark is purely algebraic, since $n$ need not be an
integer. Large-time decay properties for $U_J(t)$ (with a
different rate) are also established in \cite[Theorem~4.1]{Kim12},
but we shall not need them here. Invoking \cite[Theorem~1.2]{KT},
we infer the following lemma.
\begin{lemma}[Local Strichartz estimates]
 \label{lem:Strichartz}
Let $J\ge 2$, and $(q_1,r_1)$, $ (q_2,r_2)$ be admissible pairs,
in the sense that they satisfy
\begin{equation}
 \label{eq:adm}
 \frac{2}{q} = n\(\frac{1}{2}-\frac{1}{r}\),\quad 2\le r\le
 \infty,\quad n=\frac{3}{J}.
\end{equation}
Let $I$
be some finite time interval, of length at most one, $|I|\le 1$. \\
$1.$ There exists $C=C(r_1)$ such that for
all $\phi \in
L^2(\R^3)$,
\begin{equation}\label{eq:strich}
 \left\| U_J(\cdot)\phi \right\|_{L^{q_1}(I;L^{r_1}(\R^3))}\le C
 \|\phi \|_{L^2(\R^3)}.
 \end{equation}
$2.$ If $I$ contains the origin, $0\in I$, denote
\begin{equation*}
 D_I(f)(t,x) = \int_{I\cap\{s\le
 t\}} U_J(t-s)f(s,x)ds.
\end{equation*}
There exists $C=C(r_1,r_2)$ such that for all $f\in
L^{q'_2}(I;L^{r'_2})$,
\begin{equation}\label{eq:strichnl}
 \left\lVert D_I(f)
 \right\rVert_{L^{q_1}(I;L^{r_1}(\R^3))}\le C \left\lVert
 f\right\rVert_{L^{q'_2}\(I;L^{r'_2}(\R^3)\)},
\end{equation}
where $p'$ stands for the H\"older conjugate exponent of $p$, $\dis
\frac{1}{p}+\frac{1}{p'}=1$.
\end{lemma}
\begin{remark}
 The value $r = \infty$ is always allowed in the present
 context, because we morally consider a Schr\"odinger equation in
 space dimension $n<2$.
\end{remark}
\subsubsection{The fixed-point argument}
\label{sec:fixed}
In order to unify the treatment of the terms $H$ and $F$ in
\eqref{eq:duhamel1}, consider the trilinear operator
\begin{equation*}
 {\mathbf T}(\phi_1,\phi_2,\phi_3) = \(\frac{1}{|x|}\ast
 \(\phi_1\phi_2\)\)\phi_3.
\end{equation*}
\begin{lemma}\label{lem:trilin}
 There exists $C>0$ such that for all $\phi_1,\phi_2,\phi_3\in
C_0^\infty(\R^3)$,
 \begin{equation*}
 \|{\mathbf T}(\phi_1,\phi_2,\phi_3)\|_{L^2(\R^3)}\le C
\lVert \phi_1\rVert_{L^{24/11}(\R^3)}\lVert \phi_2\rVert_{L^{24/11}(\R^3)}
\lVert \phi_3\rVert_{L^4(\R^3)}.
 \end{equation*}
\end{lemma}
\begin{proof}
The H\"older inequality yields
 \begin{equation*}
 \|{\mathbf T}(\phi_1,\phi_2,\phi_3)\|_{L^2(\R^3)}\le
\left\|\frac{1}{|x|}\ast
 \(\phi_1\phi_2\)\right\|_{L^4(\R^3)}
\|\phi_3\|_{L^4(\R^3)}.
 \end{equation*}
Since $x\in \R^3$, the Hardy--Littlewood--Sobolev inequality (see,
e.g., \cite{BCD11}) yields
\begin{equation*}
 \left\|\frac{1}{|x|}\ast
 \(\phi_1\phi_2\)\right\|_{L^4(\R^3)}\le C \|\phi_1\phi_2\|_{L^{12/11}(\R^3)},
\end{equation*}
and the lemma follows from the H\"older inequality.
\end{proof}
\begin{lemma}\label{lem:local1}
 Let $\psi_{01},\dots,\psi_{0N}\in L^2(\R^3)$. There exists $T>0$
 depending on $\psi_{01},\dots,\psi_{0N}$ only through
 $\|\psi_{01}\|_{L^2},\dots , \|\psi_{0N}\|_{L^2}$ such that
 \eqref{eq:duhamel1} has a unique solution
 \begin{equation*}
 (\psi_1,\dots,\psi_N)\in \(C([0,T];L^2(\R^3))\cap L^{4J/3}\([0,T]
;L^\infty(\R^3)\)\)^N.
 \end{equation*}
\end{lemma}
\begin{proof}
Denote by $\Phi_k(\psi_1,\dots,\psi_L)$ the right-hand side of
\eqref{eq:duhamel1}, and for $T>0$, let
\begin{align*}
 X_T=\{ & (\psi_1,\dots,\psi_N) \in L^\infty([0,T];L^2(\R^3))^N ;\quad
 \|\psi_k\|_{L^\infty([0,T];L^2(\R^3))}\le 2\|\psi_{0k}\|_{L^2},\\
&
 \|\psi_k\|_{L^{4J/3}([0,T];L^\infty(\R^3))}\le
 2C_\infty\|\psi_{0k}\|_{L^2},\quad k=1,\dots,N\},
\end{align*}
where the constant $C_\infty$ stems from \eqref{eq:strich} in the case
$r_1=\infty$.
 The lemma follows from a standard fixed-point argument: for $T>0$
 sufficiently small (depending on $\|\psi_{01}\|_{L^2},\dots ,
 \|\psi_{0N}\|_{L^2}$), all the $\Phi_k$'s leave $X_T$ invariant, and
 are contractions on that space.

From Lemma~\ref{lem:Strichartz}, and denoting by
$L^q_TL^r=L^q([0,T];L^r(\R^3))$, we have
\begin{equation*}
 \|\Phi_k\|_{L^\infty_TL^2} \le \|\psi_{0k}\|_{L^2} + C
 \|H\psi_k\|_{L^1_TL^2} + C \|F_k (\psi_k)\|_{L^1_TL^2} + \|V\psi_k\|_{L^1_TL^2}.
\end{equation*}
Lemma~\ref{lem:trilin} and the boundedness of $V$ yield
\begin{align*}
 \|\Phi_k\|_{L^\infty_TL^2} &\le \|\psi_{0k}\|_{L^2} +
 C\sum_{\ell=1}^N\Big\|\|\psi_k(t)\|_{L^{24/11}}
 \|\psi_\ell(t)\|_{L^{24/11}}
\|\psi_\ell(t)\|_{L^4}\Big\|_{L^1_T}\\
&\quad + C\| \psi_k\|_{L^1_TL^2}.
\end{align*}
The last term is readily estimated by $CT \|
\psi_k\|_{L^\infty_TL^2}$. Each term of the sum is controlled by
\begin{equation*}
 \Big\|\|\psi_k(t)\|_{L^\infty}^{1/12}\|\psi_k(t)\|_{L^2}^{11/12}
\|\psi_\ell(t)\|_{L^\infty}^{1/12}\|\psi_\ell(t)\|_{L^2}^{11/12}
\|\psi_\ell(t)\|_{L^\infty}^{1/2}\|\psi_\ell(t)\|_{L^2}^{1/2}
\Big\|_{L^1_T}.
\end{equation*}
Neglecting the indices $k$ and $\ell$, which are irrelevant at
this step of the analysis, the H\"older inequality in time yields
\begin{equation*}
\Big\|\|\psi(t)\|_{L^\infty}^{2/3}\|\psi(t)\|_{L^2}^{7/3}
\Big\|_{L^1_T} \le \|\psi\|_{L^\infty_TL^2}^{7/3} \|\psi
\|_{L^{4J/3}_TL^\infty}^{2/3} T^{(2J-1)/(2J)},
\end{equation*}
and we come up with an estimate of the form
\begin{equation*}
\|\Phi_k\|_{L^\infty_TL^2} \le \|\psi_{0k}\|_{L^2} +
C\(\|\psi_{01}\|_{L^2},\dots,\|\psi_{0N}\|_{L^2}\) \(T + T^{(2J-1)/(2J)}\).
\end{equation*}
Choosing $T>0$ sufficiently small, the right-hand side does not
exceed $2\|\psi_{0k}\|_{L^2}$, uniformly in $k$. Similarly,
Lemma~\ref{lem:Strichartz} yields
\begin{equation*}
 \|\Phi_k\|_{L^{4J/3}_TL^\infty} \le C_\infty \|\psi_{0k}\|_{L^2} + C
 \|H\psi_k\|_{L^1_TL^2} + C \|F_k (\psi_k)\|_{L^1_TL^2} +C \|V\psi_k\|_{L^1_TL^2},
\end{equation*}
and so, $X_T$ is invariant under the action of $\Phi$ provided that
$T>0$ is sufficiently small.

Up to diminishing $T$, contraction follows readily, since $\mathbf
T$ is a trilinear operator. So, there exists a unique (in $X_T$)
fixed point for $\Phi$, that is, a solution to
\eqref{eq:duhamel1}. Uniqueness in the larger space
$\(C([0,T];L^2(\R^3))\cap L^{4J/3}\([0,T] ;L^\infty(\R^3)\)\)^N$
follows from the same estimates.
\end{proof}
Since the $L^2$ norm of $\psi_k$, $k=1,\dots,N$ is invariant under the
flow of \eqref{HOHF} (like in the case $J=1$), the above local
existence result can be iterated indefinitely in order to cover any
arbitrary time interval, and Theorem~\ref{theo:cauchyNL1} follows.

\subsubsection{Higher-order regularity}

We infer the propagation of higher-order Sobolev regularity, which
essentially reflects the fact that \eqref{HOHF} is
$L^2$-subcritical, and the nonlinearity is smooth. Roughly
speaking, the point is to differentiate \eqref{eq:duhamel1} with
respect to the space variable (such derivatives commute with
$U_J$), and use the fact that the nonlinearity is a trilinear
operator, along with Sobolev embedding.
\begin{corollary}\label{cor:higher}
 Let $s\in \N$. Suppose that $V\in W^{s,\infty}(\R^3)$, and that
 $\psi_{0k}\in H^s(\R^3)$, $k=1,\dots,N$. Then the solution to
 \eqref{HOHF} provided by Theorem~\ref{theo:cauchyNL1} satisfies
 \begin{equation*}
 \psi_k\in C(\R;H^s(\R^3)),\quad
 k=1,\dots,N.
 \end{equation*}
If $s\ge J$, then we have in addition:
\begin{equation*}
 \psi_k\in L^\infty(\R;H^{J}(\R^3)),\quad
 k=1,\dots,N.
 \end{equation*}
\end{corollary}
\begin{proof}
We refer to the proof of \cite[Theorem~8.1]{CastoM3AS97} for
precise details concerning the proof of the first statement. In
the case $s\ge J$, we take advantage of the Hamiltonian structure
of \eqref{HOHF}. The quantity
\begin{equation}\label{eq:EHF}
\begin{aligned}
 {\mathcal E}_{\rm HF}& = \sum_{k=1}^N \<\psi_k,{\mathcal H}_{0J}\psi_k\>
 +\int_{\R^3}V(x)\rho_\Psi(x)dx\\
&\quad +
 \frac{\kappa}{2} \iint_{\R^3\times\R^3}
 \frac{\rho_\Psi(x)\rho_\Psi(y)-|\rho_\Psi(x,y)|^2}{|x-y|}dxdy,
\end{aligned}
\end{equation}
is formally independent of time, where
\begin{equation}\label{eq:rho}
 \rho_\Psi(x,y) = \sum_{k=1}^N \psi_k(x)\overline{\psi_k}(y),\text{
 and }
 \rho_\Psi(x)=\rho_\Psi(x,x),
\end{equation}
and we recall that we have denoted
\begin{equation*}
 \mathcal{H}_{0J}=-\sum_{j=0}^J
\frac{\alpha(j)\hbar^{2j}}{m^{2j-1}c^{2j-2}} \Delta^j
=-\sum_{j=0}^J (-1)^j\frac{\alpha(j)\hbar^{2j}}{m^{2j-1}c^{2j-2}}
(-\Delta)^j ,
\end{equation*}
where the last equality stresses the fact that $-\Delta$ is a
positive operator. In view of the Cauchy--Schwarz inequality, the
integral on $\R^3\times \R^3$ in ${\mathcal E}_{\rm HF}$ is
nonnegative. At leading order (in terms of regularity),
\begin{equation*}
 \<\psi_k,{\mathcal H}_{0J}\psi_k\> =
 (-1)^{J+1}\frac{\alpha(J)\hbar^{2J}}{m^{2J-1}c^{2J-2}}
 \|(-\Delta)^{J/2}\psi_k\|_{L^2}^2+\text{ l.o.t.}
\end{equation*}
In view of the conservation of the $L^2$ norm, we infer that if
$(-1)^{J+1}$ and $\kappa$ have the same sign, then the conservation of
${\mathcal E}_{\rm HF}$ yields an a priori bound of the form
\begin{equation}\label{aprioriHJ}
 \psi_k\in L^\infty(\R;H^{J}(\R^3)),\quad k=1,\dots,N.
\end{equation}
In passing, we have used the following interpolation estimates, for
$0\le s\le J$:
\begin{equation*}
 \| (-\Delta)^{s/2}\psi\|_{L^2}\le C
 \|\psi\|_{L^2}^{1-s/J}\|(-\Delta)^{J/2}\psi\|_{L^2}^{s/J}.
\end{equation*}
If $(-1)^{J+1}$ and $\kappa$ have different signs, we recall from
 \cite{LiYa87} the estimate
 \begin{align*}
 \iint_{\R^3\times\R^3}
 \frac{\rho_\Psi(x)\rho_\Psi(y)}{|x-y|}dxdy&\le C
 \(\int_{\R^3}\rho_\Psi(x)dx\)^{2/3}
 \(\int_{\R^3}\rho_\Psi^{4/3}(x)dx\)\\
&\le C
 \sum_{k=1}^N\|\psi_k\|_{L^{8/3}(\R^3)}^{8/3}\le C
 \sum_{k=1}^N\|\psi_k\|_{H^{3/8}(\R^3)}^{8/3},
 \end{align*}
where we have used the conservation of the $L^2$-norm and Sobolev
embedding, successively. Therefore, the leading order in the
``kinetic'' part always dominates the potential part ($J>3/8$),
and \eqref{aprioriHJ} is always true. Finally, the conservation
of~${\mathcal E}_{\rm HF}$ can be rigorously established by the
following classical arguments (see, e.g., \cite{CazCourant}).
\end{proof}

To conclude this section, we sketch a more direct proof of the
above result.
\begin{proposition}\label{prop:algebre}
 Let $s\in \N$, with $s\ge J\ge 2$. Suppose that $V\in W^{s,\infty}(\R^3)$, and that
 $\psi_{0k}\in H^s(\R^3)$, $k=1,\dots,N$. Then
 \eqref{HOHF} has a unique, global, solution
 \begin{equation*}
 \psi_k\in C(\R;H^s(\R^3))\cap L^\infty(\R;H^{J}(\R^3)),\quad
 k=1,\dots,N,
 \end{equation*}
with initial data $\psi_{0k}$.
\end{proposition}
\begin{proof}[Sketch of the proof]
 Since $s\ge J\ge 2$, $H^s(\R^3)$ is a Banach algebra, continuously
 embedded into $L^p(\R^3)$ for all $p\in [2,\infty]$. We have
 \begin{equation*}
 \|V\psi\|_{H^s}\le C\|V\|_{W^{s,\infty}}\|\psi\|_{H^s},
 \end{equation*}
and the estimate of Lemma~\ref{lem:trilin} can be replaced by
\begin{equation}\label{eq:THs}
 \|{\mathbf T}(\phi_1,\phi_2,\phi_3)\|_{H^s(\R^3)}\le C
\lVert \phi_1\rVert_{H^s(\R^3)}\lVert \phi_2\rVert_{H^s(\R^3)}
\lVert \phi_3\rVert_{H^s(\R^3)}.
\end{equation}
To see this, decompose $1/|x|$ as the sum of $K_1(x)={\bf
1}_{|x|<1}/|x|\in L^1 (\R^3)$ and $K_2(x)={\bf 1}_{|x|\ge
1}/|x|\in L^\infty (\R^3)$. We have
\begin{align*}
 \|{\mathbf T}(\phi_1,\phi_2,\phi_3)\|_{H^s}&\le C
\left\lVert K_1\ast\(\phi_1\phi_2\)\right\rVert_{H^s}
\lVert \phi_3\rVert_{H^s} \\
&\quad+ C\sum_{|\beta_1|+|\beta_2|\le s}
\left\lVert K_2\ast\(\d^{\beta_1}\phi_1 \d^{\beta_2}\phi_2\)\right\rVert_{L^\infty}
\lVert \phi_3\rVert_{H^s} .
\end{align*}
 Then \eqref{eq:THs} follows from
\begin{align*}
 \|K_1\ast\(\d^{\beta_1}\phi_1\d^{\beta_2}\phi_2\)\|_{L^2} &\le
 \|K_1\|_{L^1} \|\d^{\beta_1}\phi_1\d^{\beta_2}\phi_2\|_{L^2} \le
 \|K_1\|_{L^1} \lVert \phi_1 \phi_2\rVert_{H^s},\\
 \|K_2\ast\(\d^{\beta_1}\phi_1\d^{\beta_2}\phi_2\)\|_{L^\infty} &\le
 \|K_2\|_{L^\infty}
 \|\d^{\beta_1}\phi_1\d^{\beta_2}\phi_2\|_{L^1}\\
&\le \|K_2\|_{L^\infty}
 \|\d^{\beta_1}\phi_1\|_{L^2}\|\d^{\beta_2}\phi_2\|_{L^2},
\end{align*}
and the fact that $H^s$ is an algebra. A classical fixed-point
argument yields the local existence of a solution in $H^s$ (see,
e.g., \cite{CazHar}). Global existence when $s=J$ follows from the
same arguments as in the proof of Corollary~\ref{cor:higher}: we
have an a priori estimate in $H^J(\R^3)$, hence in
$L^\infty(\R^3)$, so the solution is global in time. Propagation
of higher regularity (when $s>J$) follows easily, thanks to tame
estimates (see \cite{BCD11}).
\end{proof}

\subsection{Coulomb potential}\label{Sec HF Coulomb}
\label{sec:cauchyC}

In the case where the external potential $V$ in \eqref{HOHF} is a
Coulomb potential \eqref{eq:coulombpot}, we prove:
\begin{theorem}\label{theo:r3coulomb}
 Let $J\ge 2$, $V$ given by \eqref{eq:coulombpot},
 $\psi_{01},\dots,\psi_{0N}\in
 L^2(\R^3)$. Then \eqref{eq:duhamel1} has a unique, global, solution
 \begin{equation*}
 (\psi_1,\dots,\psi_N)\in \(C(\R;L^2(\R^3))\cap L^{4J/3}_{\rm
 loc}\(\R;L^\infty(\R^3)\)\)^N.
 \end{equation*}
In addition, the following conservation laws hold, for all
$\ell,k\in\{1,\dots,N\}$:
\begin{equation}\label{eq:conservL2bis}
 \frac{d}{dt}\int_{\R^3}\overline{\psi_\ell} (t,x)\psi_k(t,x)dx=0.
\end{equation}
If moreover $\psi_{01},\dots,\psi_{0N}\in
 H^{2J}(\R^3)$, then
\begin{equation*}
 \psi_k\in C(\R; H^{2J}(\R^3))\cap L^\infty(\R;H^{J}(\R^3)),\quad
 k=1,\dots,N,
 \end{equation*}
and the energy
\begin{align*}
 {\mathcal E}_{\rm HF}& = \sum_{k=1}^N \<\psi_k,{\mathcal H}_{0J}\psi_k\>
 +\int_{\R^3}V(x)\rho_\Psi(x)dx\\
&\quad +
 \frac{\kappa}{2} \iint_{\R^3\times\R^3}
 \frac{\rho_\Psi(x)\rho_\Psi(y)-|\rho_\Psi(x,y)|^2}{|x-y|}dxdy
\end{align*}
is independent of time, where $\rho_\Psi$ is defined in
\eqref{eq:rho}.
\end{theorem}
\begin{proof}[Sketch of the proof]
 The global existence at the $L^2$ level follows the same lines as in
 the previous section. The only difference is that the term $V\psi_k$
 must be handled differently. Since the pair $(4J/3,\infty)$ is
 admissible, we may write
 \begin{align*}
 \|\Phi_k\|_{L^\infty_TL^2} &\le \|\psi_{0k}\|_{L^2} + C
 \|H\psi_k\|_{L^1_TL^2} + C \|F_k (\psi_k)\|_{L^1_TL^2} \\&
\quad + C\|V_1\psi_k\|_{L^{4J/(4J-3)}_TL^1}+
C\|V_2\psi_k\|_{L^1_TL^2},
\end{align*}
where we have decomposed the Coulomb potential as the sum of a
singular potential with compact support and a bounded potential,
\begin{equation*}
 V_1(x) =\frac{\alpha}{|x|}{\mathbf 1}_{|x|<1},\quad V_2(x)
 =\frac{\alpha}{|x|}{\mathbf 1}_{|x|\ge 1}.
\end{equation*}
Since $V_2\in L^\infty(\R^3)$, the last term is treated like in
the previous case. We also have, in view of the Cauchy--Schwarz
inequality (in $x$),
\begin{align*}
 \|V_1\psi_k\|_{L^{4J/(4J-3)}_TL^1}&\le
 \|V_1\|_{L^2(\R^3)}\|\psi_k\|_{L^{4J/(4J-3)}_TL^2}\\
& \le
 T^{(4J-3)/(4J)}\|V_1\|_{L^2(\R^3)}\|\psi_k\|_{L^\infty_TL^2},
\end{align*}
and we can conclude like in the proof of Lemma~\ref{lem:local1}, and
Theorem~\ref{theo:cauchyNL1}, successively, to obtain the first part
of the theorem.

For the second part, we follow the same strategy as in
\cite{Yajima87} and \cite{CaLB99}: the above fixed-point argument
can be repeated in
 \begin{align*}
 Y_T=\{ & (\psi_1,\dots,\psi_N) \in L^\infty([0,T];H^J(\R^3))^N ;\quad
 \|\psi_k\|_{L^\infty([0,T];L^2(\R^3))}\le 2\|\psi_{0k}\|_{L^2},\\
& \|\psi_k\|_{L^{4J/3}([0,T];L^\infty(\R^3))}\le
 2C_\infty\|\psi_{0k}\|_{L^2},\\
 & \|\d_t \psi_k\|_{L^\infty([0,T];L^2(\R^3))}\le 2K_{0k},\\
&
 \|\d_t \psi_k\|_{L^{4J/3}([0,T];L^\infty(\R^3))}\le
 2C_\infty K_{0k},\quad k=1,\dots,N\},
\end{align*}
where $K_{0k}$ corresponds morally to $\|\d_t \psi_{k\mid
 t=0}\|_{L^2}$. Since the time variable is characteristic, this
quantity is given by the equation, and we can take
\begin{align*}
 \hbar K_{0k}=\sum_{j=0}\frac{|\alpha(j)|\hbar^{2j}}{m^{2j-1}c^{2j-2}}
 \| \psi_{0k}\|_{\dot H^{2j}} +\|H\psi_{0k}\|_{L^2}
 +\|F_k(\psi_{0k})\|_{L^2}+\|V\psi_{0k} \|_{L^2} .
\end{align*}
The sum on the right-hand side is finite by assumption, the
nonlinear terms are finite by Sobolev embedding, and the last term
is controlled by $\|\nabla \psi_{0k}\|_{L^2}$ thanks to the Hardy
inequality (see, e.g., \cite{BCD11}).

The fixed-point argument performed in $X_T$ is readily adapted to
the case of $Y_T$, hence
\begin{equation*}
 \psi_k\in C([0,T]; H^{2J}(\R^3)),\quad
 k=1,\dots,N.
 \end{equation*}
Since $T$ depends on the $L^2$ norms of the initial data, and not
on higher-order norms, this local argument can be repeated in
order to cover any given time interval, hence
\begin{equation*}
 \psi_k\in C(\R; H^{2J}(\R^3)),\quad
 k=1,\dots,N.
 \end{equation*}
The conservation of the energy ${\mathcal E}_{\rm HF}$ follows from
standard arguments, and
the proof of Corollary~\ref{cor:higher} can be repeated to obtain
 the global boundedness of the $H^J$ norm.
\end{proof}

\section{Ground state}
\label{sec:ground}

The existence of a ground state for the semirelativistic Hartree
equation goes
back to \cite{LiYa87}. See also \cite{CiSe15} for the introduction of
an external potential as well as for more references.

The problem of the existence of a ground state for the
higher-order Schr\"odinger equations has been raised in
\cite{LuSc14} (see also \cite{LuSc-p}). In particular, the
second-order Schr\"odinger equation ($J=2$) has no ground state.
We will see in this section that the odd-order Schr\"odinger and
Hartree--Fock equations have a ground state and are relevant in
quantum physics.

\subsection{Higher-order Schr\"odinger equation}
\label{sec:high-order-schr}

In the case of the higher-order Schr\"o\-dinger equation with a
potential, \eqref{eq:pot}, the associated energy reads
\begin{equation}
 \label{eq:nrjS}
 {\mathcal E}_{\rm S} = \<\psi,{\mathcal H}_{0J}\psi\>
 +\int_{\R^3}V(x)|\psi(x)|^2dx.
\end{equation}
Integrations by parts show that the energy also takes the form
\begin{equation}\label{eq:nrjS2}
 {\mathcal E}_{\rm S} = \sum_{j=0}^J
\frac{\alpha(j)\hbar^{2j}}{m^{2j-1}c^{2j-2}} (-1)^{j+1}\int_{\R^3}
\left|(-\Delta)^{j/2}\psi(x)\right|^2 dx +\int_{\R^3}V(x)|\psi(x)|^2dx.
\end{equation}
Let
\begin{equation}\label{eq:m}
 \underline{m} = \inf \{ {\mathcal E}_{\rm S}\ ;\ \<\psi,\psi\>=1\}.
\end{equation}
The main cases we are interested in are when $V$ is a
harmonic-oscillator potential~or when $V$ is a Coulomb potential.
In the first case, we have $V\ge 0$: in view of \eqref{eq:nrjS2},
we readily see that $\underline{m}$ is finite if and only if $J$
is odd ($\underline{m}=-\infty$ if $J$ is even). Recall that we
have seen in \cite{CaMo12} and in
Section~\ref{sec:harmonic-potential} that for $J$ even, the
dynamics associated to \eqref{eq:pot} is not well-defined when $V$
is a harmonic-oscillator potential. Consequently, the case $J$ odd
seems to be the only reliable one. When $V$ is a Coulomb potential
\eqref{eq:coulombpot}, the dynamics associated to
\eqref{eq:coulomb} is well-defined for all $J\in \N$, as stated in
Theorem~\ref{theo:coulomb}. In view of the Cauchy--Schwarz and
Hardy inequalities,
\begin{equation*}
\left| \int_{\R^3}V(x)|\psi(x)|^2dx\right|\le |\alpha| \left\lVert
 \frac{\psi}{|x|}\right\rVert_{L^2(\R^3)} \left\lVert
 \psi\right\rVert_{L^2(\R^3)} \le C |\alpha| \|\nabla
\psi\|_{L^2(\R^3)}\|\psi\|_{L^2(\R^3)}.
\end{equation*}
On the other hand, if $J$ is even, and unlike what happens when
$V$ is a harmonic-oscillator potential, one can consider
\begin{equation*}
 M = \sup \{ {\mathcal E}_{\rm S}\ ;\ \<\psi,\psi\>=1\}=-\inf \{ -
 {\mathcal E}_{\rm S}\ ;\ \<\psi,\psi\>=1\},
\end{equation*}
which is finite, for the same reason by which $\underline{m}$ is
finite when $J$ is odd. It is then classical to infer (see, e.g.,
\cite{LiSi77,Lio87}):

\begin{proposition}
 Suppose that $J$ is odd, and that $V$ is the sum of a
 harmonic-oscillator potential and a Coulomb potential,
 \begin{equation*}
 V(x) = \frac{\alpha}{|x|}+\sum_{j=1}^3\om_j^2 x_j^2,\quad \alpha
 \in \R, \ \om_j\ge 0.
 \end{equation*}
 Then there exists $\psi$ such that $\<\psi,\psi\>=1$ and $
{\mathcal E}_{\rm S}=\underline{m}$, with $\underline{m}$ defined
in \eqref{eq:m}. If $J$ is even and $V$ is a Coulomb potential
($\om_j=0$ for all $j$), there exists $\psi$ such that
$\<\psi,\psi\>=1$ and $ {\mathcal E}_{\rm S}=M$.
\end{proposition}

\subsection{Higher-order Hartree--Fock equation}
\label{sec:higher-order-hartree}

In the case of the higher-order Hartree--Fock equation
\eqref{HOHF}, the associated energy is given by \eqref{eq:EHF}.
When $J=1$ (classical Hartree--Fock equation), the existence of
minimizers for $\mathcal E_{\rm HF}$ and their properties have
been studied in, e.g., \cite{LiSi77,Lio87,Sol91, Lew11}. As we
have seen in Section~\ref{sec:cauchyB}, $\mathcal E_{\rm HF}$
controls the $H^J$-norm, so for $J\ge 2$, the Hartree nonlinearity
plays a weaker role in the analysis compared to the standard case
$J=1$. If a harmonic confinement is present in all three spatial
directions,
\begin{equation*}
 V(x) = \frac{\alpha}{|x|}+\sum_{j=1}^3\om_j^2 x_j^2,\quad \alpha
 \in \R, \ \om_j> 0,
 \end{equation*}
then any minimizing sequence is compact, since the embedding
\begin{equation*}
 H^J(\R^3)\cap \F(H^1)\hookrightarrow L^2(\R^3)\cap L^6(\R^3)
\end{equation*}
is compact (see, e.g., \cite{ReedSimon2}).
\begin{proposition}
 Let $N\ge 1$, $\kappa =1$. Suppose that $J\ge 3$ is odd, and that $V$ is the
 sum of a harmonic-oscillator potential and a Coulomb potential,
 \begin{equation*}
 V(x) = \frac{\alpha}{|x|}+\sum_{j=1}^3\om_j^2 x_j^2,\quad \alpha
 \in \R, \ \om_j\ge 0.
 \end{equation*}
In either of the cases,
\begin{itemize}
\item full confinement: $\om_j>0$, $\forall j=1,\dots,3$, or
\item $\alpha>N-1$,
\end{itemize}
 there exists $\psi\in H^J(\R^3)^N$ such that
 \begin{equation*}
 \mathcal E_{\rm HF}(\psi) = \min \left\{ {\mathcal E}_{\rm HF}(\phi),\
 \phi\in H^J(\R^3)^N\ ;\ \int_{\R^3}\phi_j\bar \phi_k =\delta_{jk}\right\}.
 \end{equation*}
If $J\ge 2$ is even, $V$ is a Coulomb potential ($\om_j=0$ for all
$j$) and $\alpha>N-1$, there exists $\psi\in H^J(\R^3)^N$ such
that
 \begin{equation*}
 \mathcal E_{\rm HF}(\psi) = \max \left\{ {\mathcal E}_{\rm HF}(\phi),\
 \phi\in H^J(\R^3)^N\ ;\ \int_{\R^3}\phi_j\bar \phi_k =\delta_{jk}\right\}.
 \end{equation*}
\end{proposition}

\section{Conclusion}\label{Sec Conclusion}
In this article, we have shown that the higher-order
Schr\"{o}dinger equations are compatible with both the
harmonic-oscillator potential and the Coulomb potential.
Moreover, we have expanded the scope to higher-order Hartree--Fock
equations~with bounded and Coulomb potentials, which may become a
useful tool in many-particle physics. Finally, we have proved the
existence of a ground state for the odd-order ones among both
types of equations, which thus are the only ones to have a
physical meaning.

\appendix

\section{Convergence of the higher-order Schr\"odinger equation
 without potential}

Recall that for $s\in \N$, the (semi-classical) Sobolev space
$H^s(\R^d)$ is the space
of $L^2$ functions whose distributional derivatives of order at most
$s$ are in $L^2(\R^d)$. It is equipped with the norm
\begin{equation*}
 \|f\|_{H^s}=\sum_{|\alpha|\le s}\hbar^{|\alpha|}\|\d^\alpha f\|_{L^2(\R^d)}.
\end{equation*}
We denote by $H^\infty(\R^d)$ the intersection of all the spaces
$H^s(\R^d)$, $s\in \N$. These spaces can also be characterized in
terms of their Fourier transform, as defined in
Section~\ref{sec:linear}. For $s\in \N$, we have the equivalence
of norms:
\begin{equation}
 \|f\|_{H^s}^2\approx\int_{\pp\in \R^d} \(1+|\pp|^2\)^s
 \left\lvert \widehat f(\pp)\right\rvert^2\dd \pp.
\end{equation}
The \emph{homogeneous} Sobolev space $\dot H^s(\R^d)$ is equipped with
the norm
\begin{equation*}
 \|f\|_{\dot H^s} = \sum_{|\alpha|= s}\hbar^{|\alpha|}\|\d^\alpha
 f\|_{L^2(\R^d)}\approx\(\int_{\pp\in \R^d} |\pp|^{2s}
 \left\lvert \widehat f(\pp)\right\rvert^2\dd \pp\)^{1/2}.
\end{equation*}

\begin{theorem}\label{th:cv}
 Let $\psi_0\in H^\infty(\R^d)$, and consider the solutions $\psi$ and
 $\psi_J$ to \eqref{SR} and \eqref{eq:pot}, respectively, in the case
 $V=0$.
Suppose that $\psi_{\mid t=0}=\psi_{J\mid t=0}=\psi_0$. Then for
all $T>0$,
\begin{equation}\label{eq:reste}
 \sup_{t\in [0,T]}\|\psi(t,\cdot)-\psi_J(t,\cdot)\|_{L^2(\R^d)}\le
 \frac{2T}{\hbar}\frac{\alpha(J+1)}{m^{2J+1}c^{2J}}
 \|\psi_0\|_{\dot H^{2J+2}(\R^d)}.
\end{equation}
In particular, if there exists $C_0$ independent of $s\in \N$ such that
\begin{equation}\label{eq:borne-unif}
 \|\psi_0\|_{\dot H^{s}(\R^d)}\le C_0 (mc)^{s},
\end{equation}
then, by \eqref{eq:stirling},
\begin{equation*}
 \sup_{t\in
 [0,T]}\|\psi(t)-\psi_J(t)\|_{L^2(\R^d)}=
\O\(T\alpha(J+1)\)
=\O\(\frac{T}{J^{3/2}}\)\Tend J \infty 0.
\end{equation*}
\end{theorem}
In \cite{CaMo12}, the above convergence result was proven under the
assumption that the Fourier transform of $\psi_0$ is supported in the
ball of radius $mc$, a case where \eqref{eq:borne-unif} becomes
trivial, since
\begin{align*}
 \|\psi_0\|_{\dot H^{s}(\R^d)}^2& \approx \int_{\pp\in \R^d} |\pp|^{2s}
 \left\lvert \widehat \psi_0(\pp)\right\rvert^2\dd \pp=
 \int_{|\pp|\le mc} |\pp|^{2s}
 \left\lvert \widehat \psi_0(\pp)\right\rvert^2\dd \pp\\
&\le
(mc)^{2s}\int_{|\pp|\le mc} \left\lvert \widehat
\psi_0(\pp)\right\rvert^2\dd \pp = (mc)^{2s}\|\psi_0\|_{L^2}^2.
\end{align*}
The present extension is valid also for Gaussian wave packets
\begin{equation*}
 \psi_0(x) = \(\frac{mc}{\hbar}\)^{3/4}e^{-mc|x|^2/\hbar},
\end{equation*}
a case which was not covered in \cite{CaMo12}.
\begin{proof}
The Taylor formula yields
 \begin{equation*}
 E-E_J= mc^2 \(\frac{\pp}{mc}\)^{2J+2}\frac{1}{(J+1)!}\int_0^1
 f_{J+1}\(\theta \(\frac{\pp}{mc}\)^2\)(1-\theta)^{J}\dd\theta,
 \end{equation*}
where
\begin{equation*}
 f_n(x)= \frac{d^n}{dx^n}\(\sqrt{1+x}\).
\end{equation*}
We infer
\begin{equation}\label{eq:resteFourier}
 |E-E_J|\le mc^2 \(\frac{\pp}{mc}\)^{2J+2}\alpha(J+1).
\end{equation}
 The functions $\psi$ and $\psi_J$ solve, respectively,
 \begin{equation*}
 i\hbar \frac{\d \psi}{\d t} = E(-i\hbar\nabla_x)\psi;\quad
 i\hbar \frac{\d \psi_J}{\d t} = E_J(-i\hbar\nabla_x)\psi_J.
 \end{equation*}
The difference $w_J=\psi-\psi_J$ satisfies $w_{J\mid t=0}=0$ and solves
\begin{equation*}
 i\hbar \frac{\d w_J}{\d t} =E(-i\hbar\nabla_x)w_J +r_J,\quad \text{where
 }r_J = \(E(-i\hbar\nabla_x)-E_J(-i\hbar\nabla_x)\)\psi_J.
\end{equation*}
Multiply the above equation by $\overline{w_J}$, integrate in space,
and take the imaginary part: the term involving $E(-i\hbar\nabla_x)w_J$
disappears (because it is real), and we infer
\begin{equation}\label{eq:estnrj}
 \|w_J(t)\|_{L^2}\le \frac{2}{\hbar}\int_0^t\|r_J(\tau)\|_{L^2}\dd
 \tau.
\end{equation}
In view of the Plancherel formula, and since $E$ and $E_J$ are
Fourier multipliers,
\begin{equation*}
 \|r_J(\tau)\|_{L^2} = \|\widehat r_J(\tau)\|_{L^2}=\left\lVert
 \(E(\pp) - E_J(\pp) \)\widehat \psi_J(\tau)\right\rVert_{L^2}.
\end{equation*}
As noticed in \cite{CaMo12}, we have the explicit formula (since
$E_J$ is a Fourier multiplier)
\begin{equation*}
 \hat\psi_J(t,\pp)=\hat\psi_0(p)e^{it E_J(p)},
\end{equation*}
hence
\begin{equation*}
 \|r_J(\tau)\|_{L^2} = \left\lVert
 \(E(\pp) - E_J(\pp) \)\widehat \psi_0(\tau)\right\rVert_{L^2}.
\end{equation*}
The inequality \eqref{eq:resteFourier} and the Plancherel formula
yield
\begin{equation*}
 \|r_J(\tau)\|_{L^2}\le mc^2 \(\frac{1}{mc}\)^{2J+2}\|\psi_0\|_{\dot H^{2J+2}},
\end{equation*}
and the result follows (use \eqref{eq:stirling} for the final equality).
\end{proof}

\bibliographystyle{unsrt}
\bibliography{hartree}

\end{document}